%% file: R1_RevisedManuscript.tex
\documentclass[10pt,twocolumn,twoside]{IEEEtran}
\usepackage[utf8]{inputenc}
\usepackage{setspace}

\usepackage[acronym]{glossaries}
\usepackage{subcaption}  

\usepackage{tikz}
\usepackage{pgfplots}
\usepgfplotslibrary{groupplots}
\pgfplotsset{compat=newest}
\usepgfplotslibrary{colorbrewer, groupplots, patchplots}
\pgfplotsset{
    compat=newest,
    legend style={font=\footnotesize, fill opacity=0.7,  draw opacity=1, text opacity=1, draw=white!15!black, legend cell align=left, align=left}, 
    width=0.8\columnwidth, 
    scale only axis,
    height=4cm,
    yminorticks=false,
    xminorticks=false,
    label style={font=\small},
    title style={font=\small},
    tick align=outside,
    tick pos=left,
    tick style={color=black},
    tick label style={font=\footnotesize},
    grid style={line width=.1pt, draw=gray!20},
    major grid style={line width=.1pt,draw=gray!20},
    plot coordinates/math parser=false 
}

\usetikzlibrary{shapes,positioning}
\usepackage{amsmath}
\usepackage{amsthm}
\usepackage{amssymb}
\usepackage{amsfonts}
\usepackage{dsfont}
\usepackage{algorithm}
\usepackage[noend]{algpseudocode}
\usepackage[subnum]{cases}
\usepackage[nocomma]{optidef}
\usepackage[normalem]{ulem}
\usepackage{mathtools}
\DeclarePairedDelimiter\ceil{\lceil}{\rceil}

\DeclareMathAlphabet{\mathcal}{OMS}{cmsy}{m}{n}
\SetMathAlphabet{\mathcal}{bold}{OMS}{cmsy}{b}{n}

\newtheorem{theorem}{Theorem}
\newtheorem{remark}{Remark}
\input{acronym}

\newcommand{\rv}[1]{\textcolor{black}{#1}}

\title{Time-constrained Federated Learning (FL) in Push-Pull IoT Wireless Access}
\author{Van Phuc Bui, Junya Shiraishi, Petar Popovski, Shashi Raj Pandey%
\\
Department of Electronic Systems, Aalborg University, Denmark\\
Emails: \{vpb, jush, petarp, srp\}@es.aau.dk\thanks{This work was supported partly by the Villum Investigator Grant ``WATER" from the Velux Foundation, Denmark, and partly by the Horizon Europe SNS ``6G-XCEL" project with Grant 101139194. The work of J. Shiraishi was supported by Horizon Europe Marie Sk{\l}odowska-Curie Action (MSCA) Postdoc Fellowships with grant No~101151067.}}
\begin{document}
\maketitle
\begin{abstract}
Training a high-quality Federated Learning (FL) model at the network edge is challenged by limited transmission resources. Although various device scheduling strategies have been proposed, it remains unclear how scheduling decisions affect the FL model performance under temporal constraints. This is pronounced when the wireless medium is shared to enable the participation of heterogeneous Internet of Things (IoT) devices with distinct communication modes: (1) a scheduling (\emph{pull}) scheme, that selects devices with valuable updates, and (2) random access (\emph{push}), in which interested devices transmit model parameters. This work investigates the interplay of push-pull interactions in a time-constrained FL setting, where the communication opportunities are finite, with a utility-based analytical model. Using real-world datasets, we provide a performance tradeoff analysis that validates the significance of strategic device scheduling under push-pull wireless access for several practical settings. The simulation results elucidate the impact of the device sampling strategy on learning efficiency under timing constraints.  
\end{abstract}
\vspace{-5pt}
\section{Introduction}
\gls{fl}~\cite{mcmahan2017communication} leverages distributed data and decentralized computing to train a learning model without exchanging the raw data. FL involves a tightly coupled iterative process, where the devices undergo local training and exchange the updates for model aggregation through frequent communication with a \gls{ps}. The privacy-preserving feature of \gls{fl} is desirable in a variety of applications, i.e., edge intelligence \cite{nguyen2021federated}, semantic communications \cite{kountouris2021semantics}, or fast inference capabilities for downstream tasks. 

The device selection problem and the communication bottleneck pose the most significant challenges in \gls{fl}, particularly when the communication resources are shared amongst several resource-constraints devices~\cite{kairouz2021advances}. 
In order to realize the improvement of the model accuracy under different constraints due to system and data level heterogeneity, \cite{singhal2024greedy} introduced approaches for \emph{data valuation} and \emph{strategic sampling}. 
This enables the \gls{ps} to strategically select the subsets of \glspl{ue} having useful updates. 
However, doing \gls{fl} in the heterogeneous \gls{iot} networks further adds unique challenges and constraints, which are coupled due to the availability of communication resources, and heterogeneity in \gls{ue} computing capability and data quality~\cite{kairouz2021advances, singhal2024greedy}. The data valuation approach considered in~\cite{singhal2024greedy} brings two issues: \emph{(i)} the \gls{ps} cannot deal with stragglers, which are the devices with slow computation capabilities~\cite{reisizadeh2022straggler}, without extending the periodicity of aggregation. This impacts the FL training process and timely inference to support time-critical applications at the network edge, such as autonomous driving and industrial automation \cite{zhang2022federated}; \emph{(ii)} the distribution of the local training latency is unknown, which makes the whole training process intractable. 
The work in~\cite{mcmahan2017communication, pandey2020crowdsourcing} applied a synchronous model update method, in which the \gls{ps} waits for the slowest \gls{ue} to complete its local training and transmit the model parameters, which is seen as the worst case scenario. Therefore, this approach can not be applicable for time-shared \gls{fl} systems considered here, as this only extends the model aggregation period.

We address these issues by integrating pull- and push-based communication~\cite{cavallero2024coexistence} for a \gls{fl} setup. This enables the \gls{ps} to aggregate the local update transmitted in a push-based manner. In addition, such a communication paradigm allows the \gls{ps} to directly ask the model parameter from the subset of \glspl{ue}, which is likely to contribute to improving the global model accuracy through a pull-based approach~\cite{singhal2024greedy}. 
However, this also includes unique challenges caused by the nature of \gls{ra} for the model update based on the push-based communication. Specifically, the transmission success probability of local update in the push-based communication affects PS's scheduling decisions for the \glspl{ue} during the pull-based communication period. The novelty in this work is the strategic aspect of push-pull access of \gls{iot} \glspl{ue}, which is based on the data valuation and aims to achieve improvement in the FL generalization performance. Our contributions are threefold: \emph{(i)} we introduce the push-pull system to the \gls{iot} \gls{fl} setup, adapting the utility-based \emph{strategic} \gls{ue} scheduling approach; \emph{(ii)} we characterize the system level performance of push-pull systems for \gls{fl} setup, in terms of training accuracy and latency cost; \emph{(iii)} using numerous experiments on real-world datasets, we validate that contribution-based device scheduling enables fast knowledge acquisition and timely inference in resource-constrained settings.  

\emph{Notation:} $[m]$ denotes a set with $m$ elements; $|[m]|$ denotes the cardinality of the set $[m]$; $\mathbb{E}[\cdot]$ denotes the expectation operator of a \gls{rv}; $\mathrm{Pr}(x)$ denotes probability of happening an event $x$; $\mathbb{R}_+$ denotes a set of non-negative real numbers; $\mathds{1}_{\{\cdot\}}$ is an indicator function.

\section{Setting and Problem Definition}\label{sec:problem_def}
We consider a scenario where $K$ \glspl{ue} and a \gls{ps} train a global \gls{fl} model in a push-pull manner.
In the \emph{pull-based} communication, the \gls{ps} schedules the data transmission timing of \glspl{ue} to collect the current local model. In the \emph{push-based} communication, \glspl{ue} transmits local updates by contending the communication channel with other ready-to-transmit \glspl{ue} for model aggregation at the \gls{ps}. These \glspl{ue} contend to upload local updates once the local training is complete.

\subsection{Standard \gls{fl} Problem}
We assume each \gls{ue} $k \in [K]$ holds $N_k$ data samples in the set $\mathcal{D}_k$ as pairs $\{z_i, y_i\}_{i=1}^{N_k}$; the pair $\{z_i, y_i\}$ indicates $i-$th input-label sample, as in a classification problem. Here,  $[K] = \{1, \ldots, K\}$ is the set of total available UEs in the system. In the \gls{fl} setting, the $
K$ \glspl{ue} collaborate to train a single learning model ${\bf{w}}^*$ at the \gls{ps} by solving the following empirical risk-minimization problem in its standard objective form:
\begin{equation}
    {\bf{w}}^* = \underset{{\bf{w}}}{\arg \min} \; F({\bf{w}}; [K]) := \underset{\bf{w}}{\arg \min} \frac{1}{K} \mathbb{E}[\sum_{k\in[K]}l_k(\bf{w})], \label{eq:fl_problem}
\end{equation}
where $\{l_k\}_{k \in [K]}$ indicates the local objective function at UE $k$, denoted as $l_k({\bf{w}}):= \frac{1}{N_k}\sum_{(z_i, y_i)\in \mathcal{D}_k}l(y_i, \mathbf{f}({\bf{w}}, z_i))$, where $l(\cdot, \cdot)$ is the loss function defined with the input to label mapping function $\mathbf{f}({\bf{w}},z_i)$. Problem \eqref{eq:fl_problem} can be solved effectively iteratively using FedAvg~\cite{mcmahan2017communication}, or employing variants of distributed optimization methods, as indicated in \cite{kairouz2021advances}, with convergence guarantees following standard assumptions on the loss functions. We make similar assumptions on the loss functions: loss functions are $L-$smooth for $L>0$, $\mu-$convex for $\mu>0$, and exhibit bounded gradients and variance, such that $||\nabla F_k({\bf{w}}) - \nabla F({\bf{w}}))||\le \delta$ and  $||\nabla F_k({\bf{w}})) - \nabla F({\bf{w}}))||^2 \le \Delta, \forall k \in [K]$. 

\subsection{Communication Model}
Similar to \cite{cavallero2024coexistence},
we assume a time-slotted framework organized into frames. Each frame comprises a single downlink transmission slot, followed by $M$ uplink time slots. The duration of each frame, denoted as $T_{\mathrm{F}}$, is considered as one global model iteration time $I_g (\theta)$, described as a function of an relative accuracy level $\theta$ attained with the FL training \cite{pandey2020crowdsourcing}. Each global round $I_g (\theta)$ includes time for local model training, transmission of model updates, and model aggregations as received by the PS. Intuitively, a high target accuracy implies more time slots needed for convergence, when feasible, which is constrained by the allowed $I_g (\theta)$ in the presence of stragglers. We denote the length of each slot as $\tau$ [s], which can be expressed as $\tau = \lfloor \frac{T_{\mathrm{F}}}{M} \rfloor$. We also consider that $K$ UEs are active at the beginning of the frame.  

Following the \gls{mac} frame structure and mode of operations~\cite{cavallero2024coexistence}, we divide a time frame $T_\textrm{F}$ into two distinctive parts: 1)  pull-based communication periods $T_\textrm{Q}$, where queried \gls{ue}s transmit their instantaneous local models based on a shared schedule; and 2) push-based communications periods $T_\textrm{C}$, where \glspl{ue} with available local models transmit data following the framed-ALOHA protocol. The pull-based communication periods comprise of exactly $Q$ uplink transmission slots; thus, $T_{\mathrm{Q}}=Q\tau$, while the push-based period comprises the remaining $ M-Q $ slots, yielding $T_{\textrm{C}} = (M-Q)\tau $. Each \gls{ue} $k \in [K]$ has a local model of size $L_k$ [bits] to transmit, which is obtained with the completion of $I_l (\epsilon) \in \mathbb{R}_+$ local iterations following an arbitrary variant of \gls{sgd} to reach a local accuracy of $\epsilon$~\cite{pandey2020crowdsourcing}.
Then, the transmission time required for $L_k$ bits by UE $k$ is $T_{k,\textrm{com}}={L_k}/{R_k}$, where the rate $R_k$ is set to $R_k = {L_k}/{\tau}$, for simplicity, so that the transmission can be adapted within $\tau$. Here, we assume that the downlink communication is error-free while the collision channel is assumed for uplink transmission, where the packet loss will happen if more than one \gls{ue} transmits updates in the same communication slot.

\subsubsection{Scheduling}\label{sec:pull_model} At the beginning of each frame, the \gls{ps} selects the subset of \glspl{ue} for scheduling in the pull-reserved periods. This work applies \emph{utility-based} node selection at the \gls{ps} to quantify $Q$, which is derived as a solution of the subset selection problem formalized in Sec.~\ref{sec:utility_based_selection}. In that sense, $Q$ can be considered the available UE scheduling budget. Then, let us define \gls{ue} selection with indicator variable $\alpha_k \in \{0,1\}, \forall k \in [K]$, where $\alpha_k = 1$ ($\alpha_k = 0$) means that the \gls{ue} $k$ is selected (not selected) for model aggregation in pull-reserved periods. The \gls{ps} strategically schedules $Q$ \glspl{ue}, collected in a set defined by $[N_Q]$ for the pull-based communication such that $\sum_{k \in [N_Q]} \alpha_k = Q$. Following the node-selection, the \gls{ps} broadcasts the available global model and the information on the timing of the push-based communication, as in~\cite{cavallero2024coexistence}.
Upon receiving the global model, scheduled \glspl{ue} undergo local training and transmit the latest model weight in their allocated slots.

\subsubsection{Random Access}\label{sec:push_model}
The \glspl{ue} that does not receive the pull request from the \gls{ps} transmit their data in a shared slot in a push-based manner upon the completion of their calculation, as in \cite{cavallero2024coexistence}. Consider $[N] \subseteq [K]\setminus [N_Q]$ be the subset of \glspl{ue} that completes its local model calculation for transmission within the push slots and participates in the global model update. Deriving $N$ requires modelling the actual distribution of latency, which is non-trivial
given it depends on the heterogeneity in the computing capability of
UEs. Hence, using an approximation approach compatible with the setting,  outlined in Sec.~\ref{sec:time_cost}, we consider $N$ \glspl{ue} transmit their updates following the framed ALOHA protocol, in which $N$ \glspl{ue} randomly selects a single slot from the frame that represents the set of available push random access slots.
\vspace{-0.16in}
\subsection{Computation Model}
For each UE $k\in [K]$, denote $D_k$ [bits] as the size of the local data, $f_k$ as the computing frequency and $c_k$ as the number of CPU cycles to process one bit of the local data, which is a stochastic parameter \cite{suman2023statistical} that reflects the variable compute abilities of UEs. \rv{Hence, we use $c_k$ to capture system-level heterogeneity in our model. For this, consider an \gls{rv} $\textbf{C}$ that follows a Gamma distribution with shape and scale parameters, respectively, $\kappa$ and $\beta$. Then, $c_k$ can be sampled from
\begin{equation}
 f_\textbf{C}(c) = \frac{1}{\beta^\kappa\Gamma(\kappa, \beta)}c^{\kappa -1}\exp(-c/\beta),   
\end{equation}
where $\Gamma(a,b) = \int_0^{b} t^{a-1}e^{-t}dt$ is the Gamma function, which is a simpler parametric model for this purpose, as used in related literature (e.g., see \cite{suman2023statistical}).} Then, the time to complete one iteration of local model training by UE $k$ is ${D_kc_k}/{f_k}$. Recall and denote $I_{k,l}(\epsilon)$ as the number of local iterations offering theoretical guarantees to converge to some fixed $\epsilon$ accuracy\footnote{We refer to \cite{pandey2020crowdsourcing} for a precise definition of local and global accuracy level, $\theta$ and $\epsilon$, respectively.}; then, the total time spent by the UE $k$ for the local model update, following arbitrary variant of SGD \cite{mcmahan2017communication}, is $T_{k, \textrm{comp}} = I_{k,l}(\epsilon){D_kc_k}/{f_k}.$ Note that the number of iterations $I_{k,l}(\epsilon)$ is the lower bound to $\mathcal{O}(\log({1}/{\epsilon}))$. This brings forward the following remark and a Theorem on latency bound. 
\begin{remark}
\rv{The latency cost $T_k$ incurred due to local computations for each \gls{ue} $k \in [K]$ depends on the training dataset size $D_k$, computing frequency $f_k$, and available CPU cycles. It is an independent continuous \gls{rv} such that $p\log\left(\frac{1}{\epsilon}\right) \leq T_k \leq 2T_{\textrm{F}}$ 
almost surely, where $p > 0$ scales the time cost of local iterations $I_{k,l}(\epsilon)$.}
\end{remark}
  \begin{theorem} (Latency Bound for Local Training) Given the set of RVs defining the incurred latency for local model training $T_1, T_2, \ldots, T_k, \forall k \in [K]$, the average local latency cost $T^l_\textrm{cost}$ for a target local-global accuracy  pairs $(\epsilon, \theta)$ convergence can be bounded as 
\begin{equation}
    \mathbb{E}[T^l_\textrm{cost}] \le \frac{1}{K}\sum_{k=1}^KT_k + \log\bigg(\frac{1}{\epsilon^2}\bigg)(2q-p(1-\theta))  \sqrt{\frac{\log(1/h)}{2}}
\end{equation}  
where $h=  \exp \Bigg(-\frac{2T_{\max}^2}{K(2T_{\mathrm{F}}-p\log(1/{\epsilon}))^2}\Bigg)$, \rv{$q > 0$ is a scaling constant associated with the number of global iteration, i.e., $I_g(\theta)$ to the target accuracy $\theta$ with $\epsilon$ local accuracy and the time-budget $T_{\max}$, and term $T_\textrm{F}$ accounts for the protocol design to accommodate $T^l_\textrm{cost} = 2T_\textrm{F}$. The scaling constants $p, q$ depends on data size and the condition number of the local learning problem at UEs~\cite{konevcny2017semi}.}
\end{theorem}
\begin{proof}
Following Remark 1, the proof can be established with Hoeffding's inequality measure \cite{boucheron2003concentration} and a bound on the minimum local iterations required to attain $\epsilon$ target accuracy. Details omitted for brevity.
\end{proof}
\subsection{Time cost}\label{sec:time_cost}
Next, we derive the time cost for each global iteration to formalize the problem. In each frame, the participation nodes for the model update spend time on communication and computation. Let $[Y]= [N_Q] \cup [N']$ be a set of UEs involved in one round of global model update through strategic pull and random push, with $N'$ UEs successful in the transmission of their local model parameters.  \rv{Here, we assume no retransmissions in case of failures such that the subsequent queuing delay is ignored.} Then, considering the frame structure, the execution time for the global iteration depends on the maximum amount of time required for all successful updates from $[Y]$ \glspl{ue}, which is proportional to $T_{\textrm{cost}}:= \max\{T_{k, \textrm{comp}} + T_{k, \textrm{com}}| k \in [Y]\}.$
The time cost for \gls{ue} $k \in [N_Q]$ depends on the scheduled slots: when it transmits $\zeta$-th pull slots, its time cost can be $T_k = \zeta +1$. On the other hand, the time cost for \gls{ue} $k \in [N^{\prime}]$, denoted as $T_k$, as follows: 
\begin{equation}
T_k=\begin{cases}
			Q +1, & \text{if}~T_{k, \textrm{comp}} \le  T_{Q} \\
          \ceil{\frac{T_{k, \textrm{comp}}}{\tau}} +1 , &  \text{otherwise.} \label{eq:dis_time}
		 \end{cases}  
\end{equation}
Then, the total latency cost of receiving $[Y]$ updates is proportional to $T_{\textrm{cost}}:= \max\{T_k| k \in [Y]\}$. Note that in practice $T_k$ has to be treated probabilistically due to the RV $c_k$.

Here, the number of successful updates from push devices $N'$ requires modeling the actual distribution of latency, which is non-trivial as it depends on the heterogeneity in computing capability of \glspl{ue}. As the focus of this paper is characterizing the system-level performance in the push-pull communication regimes in a utility-based \gls{fl} training framework, we analyze the total time cost given $N$. Given $N$ and channel $M$ slots, the probability that a single user succeeds in data transmission can be described as $p_\mathrm{s}=\left(1-\frac{1}{M-Q}\right)^{N-1}$~\cite{cavallero2024coexistence}. We then use the discrete version of order statistics to derive the expected maximum time cost, to attain the target accuracy level, in which the index of transmission slot is a random variable and each node chooses the index i.i.d. uniform distribution. \rv{Mathematically, given $Q$ slots available for pull, the expectation of the overall communication cost per communication round can be expressed as 
\begin{equation}
\Bar{T}_\textrm{cost} := \mathbb{E}[T_{f|Q}] = Q(1-p_s)^N+\sum_{i=1}^{N}(Q+\ell_{\iota})p_s(1-p_s)^{i-1}, \label{eq:timecost}
\end{equation}
where $\ell_{\iota}$ is the average slot index of $\iota$-th highest index. Here we approximately calculate the mean value as $\ell_{\iota} \approx (M-Q+1)\frac{N-\iota+1}{N+1}$, based on the order statistics of a continuous random variable. This approximation becomes exact as $M-Q$ becomes large.}
\vspace{-20pt}
\subsection{Overall problem definition}
The overall problem is a subset selection problem with constraints on the cost of learning in terms of the incurred latency cost to achieve a target accuracy level $\theta_\textrm{th}$. 
\begin{align}
	{\textbf{P:}} \min_{\boldsymbol{\alpha}, \bf{w}} \quad 
	& \Bigg[\max_{} \sum_{k \in [Q]} \ F(\bf{w}; [Q])  \Bigg],  \label{objective_func}\\
	\text{s.t.} \quad & \sum_{k=1}^K \alpha_k \le Q, \\
     & \alpha_k \in \{0,1\}, \forall k \in [K], \label{cons:participation} \\	
     & \Bar{T}_\textrm{cost} \le T_{\max}. \label{cons:latency_budget}
\end{align}
In a nutshell, the \gls{ue} selection problem $(\textbf{P})$ is challenging, particularly due to the introduction of push \gls{ra} periods, which impact the overall transmission cost for model training. Specifically, given constraints on the available channel uses $M$ per global iterations, i.e., given the latency budget for learning, the \gls{ps} aims to strategically pull the most contributing $Q$ updates in each round while assessing the impact of the \gls{ra} procedure in the overall model training. We resort to a valuation-based utility design in the PS (detailed in Sec.~\ref{sec:utility_based_selection}) to effectively solve $(\textbf{P})$ and understand the interplay between push and pull communication for training a model to a target accuracy level at the edge.

\section{Utility-based UE Scheduling Algorithms}\label{sec:utility_based_selection}
Here, we outline the procedure for strategic UE scheduling, which impacts random access policies under the push-pull coexistence regime. Push updates are aggregated at the end of each frame, while scheduling in each communication round $I_g$ is based on the valuation of received updates.
\input{algo}

\textbf{Utility-based UE selection:} Let $U(\cdot)$ denote a utility function on $2^{[M]} \to \mathbb{R}$, which associates a reward/value with every subset of clients. In principle, the utility $U(\cdot)$ takes the performance of the instantaneous model \cite{singhal2024greedy} obtained by soliciting local models from UEs in each global iteration $I_g$, at the PS. For this, the \gls{ps} uses the validation data $\mathcal{D}_{\textrm{val}}:\{z_i, y_i\}_{i=1}^{N_\textrm{val}}$.
Given [Y], the \gls{sv} of \gls{ue} $k\in[Y]$ is defined: 
\begin{equation}
    \nu_k = \frac{1}{Y}\sum_{S \in [Y]\setminus k} \frac{U(S \cup k) - U(S)}{\binom{Y-1}{|S|}}, \label{eq:sv} 
\end{equation}
where local updates from $Y$ \glspl{ue} are available at the PS to derive the marginal contributions. \rv{Even though a PS is often equipped with a high computational capability, the complexity of evaluating \eqref{eq:sv} cannot be overlooked as the number of UEs grows, leading to intractable solutions. This is particularly due to the combinatorial nature of the problem, demanding a single-round computational complexity of $\mathcal{O}(Y\log (Y))$. We utilize GTG-Shapley \cite{liu2022gtg} which reduces complexity to $\mathcal{O}(\log Y)$ by implementing the truncated Monte Carlo sampling to approximate \eqref{eq:sv} efficiently. The details are given in Algorithm.~\ref{algo:algo1}}. 

\vspace{-10pt}
\section{Simulation Results}
We have simulated the push-pull interaction learning environment on a single server with 26 core Intel Xeon 2.6 GHz, 256 GB RAM, 4 TB, Nvidia V100 GPU, Ubuntu OS, with a total of $K = 200$ \gls{ue} and conducted extensive experiments on classification tasks with well-known MNIST~\cite{lecun1998gradient} and CIFAR10 \cite{krizhevsky2009learning} datasets. Unless specified, the experiments are conducted on the MNIST dataset. A multilayer perceptron (MLP) classifier and a convolutional neural network (CNN) were respectively trained for MNIST and CIFAR-10 datasets. Similar to \cite{singhal2024greedy}, data heterogeneity is introduced by distributing training samples across clients using a Dirichlet ($\alpha$) distribution with $\alpha = 1$, ensuring a moderate skew in label distribution. Systems heterogeneity is modelled by selecting a fraction $0.5$ of clients as stragglers, which means these clients transmit partial solutions by training for a randomly chosen number of epochs in $[1, E]$.
We set $E = 200$. Privacy heterogeneity is incorporated by assigning varying noise levels to client updates, where the noise level $\sigma$ follows $\sigma_k = \frac{(k-1)\sigma}{N}$ for each client, set at $\sigma = 0.1$ to represent minimal privacy variance. We set the available time resources per frame $M = 20$. The default value of $T_{\max}$ is set $M \tau$. 
As a benchmark scheme, we apply three different model aggregation methods: 1) \textbf{FedAvg}~\cite{mcmahan2017communication}, where PS randomly selects $M$ UEs per round, disregarding the device-specific impact on model performance, 2) \textbf{GreedyShap}~\cite{singhal2024greedy}, where PS selects $M$ devices with the highest Shapley values, prioritizing devices that maximize model contribution, and 3) \textbf{Centralized}, where 10\% of device send data directly to the server for model training. 
The \textbf{Proposed} approach follows Algorithm~\ref{algo:algo1} and dedicates all slots in the first $r_{\mathrm{th}} = 20$ [frames] for pull operation to initialize the SVs; then, it reverts to the fixed 50\% pull and 50\% push split. 

\input{fig2_different_setting}


Fig.~\ref{fig:com_cost} presents the accuracy of our proposed methods under different configurations of pull/push slot allocations, along with baseline methods, \textbf{FedAvg} and \textbf{Centralized}. The range for $r < r_{\mathrm{th}}$ represents the initial phase in which all slots are allocated exclusively to the pull scheme. After this phase, i.e., $r \geq r_{\mathrm{th}}$, the system transitions to the \textit{Pull+Push} configuration, allowing certain slots to accommodate spontaneous push-based updates from UEs. \rv{The results indicate that the \textbf{Proposed} (5/10, 10/10, and 20/10) push/pull configuration achieves a competitive performance with \textbf{GreedyShap},} and higher accuracy and faster convergence than other baseline methods. In contrast, the \textbf{FedAvg}  approach shows slower convergence, while the \textbf{Centralized}, though achieving high accuracy, requires more communication rounds to stabilize. These results highlight the advantages of combining pull and push mechanisms to enhance communication efficiency and model accuracy in federated learning. 
\input{fig1_acc}
Fig.~\ref{fig_convergency} demonstrates the superior performance of the proposed method in MNIST and CIFAR-10 datasets. The \textbf{Proposed} and \textbf{GreedyShap} converge faster (around iteration 25) and offer a stable performance as compared with the baselines. In contrast, while reaching a similar final performance, the \textbf{Centralized} method converges more slowly and stabilizes only after iteration 50. As expected, the \textbf{FedAvg} method performs poorly due to its lack of structured decision-making. The proposed method, however, strikes a balance between exploration and exploitation, offering greater flexibility while maintaining high performance. Unlike \textbf{GreedyShap}, which consistently selects the best candidates based on immediate potential, however, might miss potential updates for improved generalization performance due to limited exploration. This balance is crucial for maintaining both flexibility and high learning/training efficiency.
\input{fig5_MNIST_acc_vs_no_users}
\begin{figure}[t]
    \centering
    \input{fig6_latency_ver2}
    \vspace{-0.5cm}
    \caption{Total time cost for different levels of target accuracy.}
    \label{fig:latency_cost}
    \vspace{-0.5cm}
\end{figure}

Fig.~\ref{fig:acc_vs_no_users} offers a comparative analysis of the impact of increasing the number of push \glspl{ue} $N$ on the test accuracy. As observed, increasing the number of push \glspl{ue} lowers the test accuracy due to frequent collision; however, our proposed method outperforms other baselines and offers competitive performance as compared with \textbf{GreedyShap}, which selects \glspl{ue} based on data valuation. Finally, Fig.~\ref{fig:latency_cost} presents the result of the overall time cost required, as defined mathematically in \eqref{eq:timecost}, to achieve different levels of target accuracy $\theta_{\mathrm{th}}$ for different methods. We evaluate this as the minimum number of time units (in terms of $\tau$) required to hit the target accuracy. Note that this equivalent measure of the cumulative latency cost is considered to simplify our analysis. As we do not have full statistics on the per-device computation, we cannot guarantee local training completion per frame required for further analysis of latency cost. \textbf{FedAvg} incurs the highest latency and fails to achieve accuracy above 0.9. In contrast, \textbf{Proposed} achieves high accuracy with lower latency than both \textbf{GreedyShap} and \textbf{FedAvg}, demonstrating superior efficiency. Notably, for accuracy above 0.9, \textbf{GreedyShap} experiences a sharp rise in latency, requiring significantly more time.
\vspace{-5pt}
\section{Conclusions}
This work investigated a communication-efficient FL training procedure under the push-pull communication paradigm. We considered a utility-based approach to schedule relevant model updates in the pull phase while exploiting the random access procedure in the push phase to handle diverse UE participation. An analytical model captures the protocol design choices for time-constrained FL, considering the dependency between local update success in the push phase and the PS's scheduling in the pull phase. 
Experimental analysis showed the proposed aggregation strategy obtains the target accuracy with minimal latency cost compared to baseline schemes.
\vspace{-5pt}
\setstretch{0.85}
\bibliographystyle{ieeetr}
\bibliography{netl}
\end{document}

%% file: acronym.tex
\newacronym{ap}{AP}{Access Point}
\newacronym{fl}{FL}{Federated Learning}
\newacronym{iot}{IoT}{Internet of Things}
\newacronym{mac}{MAC}{Medium Access Control}
\newacronym{ml}{ML}{Machine Learning}
\newacronym{ps}{PS}{Parameter Server}
\newacronym{ra}{RA}{Random Access}
\newacronym{rv}{RV}{Random Variable}
\newacronym{sv}{SV}{Shapley Value}
\newacronym{sgd}{SGD}{Stochastic Gradient Descent}

\newacronym{ue}{UE}{User Equipment}

%% file: algo.tex
\begin{figure}[t!]
\begin{algorithm}[H]{\fontsize{8.5pt}{8.9pt}\selectfont
\caption{Strategic FL Training using GTG-Shapley }\label{algo:algo1}
 \textbf{Input}: $K$ clients with datasets $\{z_i, y_i\}_{i=1}^{K}$, validation dataset $\mathcal{D}_{\textrm{val}}:\{z_i, y_i\}_{i=1}^{N_\textrm{val}}$, initial model weight ${\bf{w}}^{(0)}$, number of channel uses $M$, time budget $T_{\max}$, UE scheduling strategy $Q$, exponential rate $\zeta$.\\
\textbf{Hyperparameters:} Training epochs per round $E$, mini-batches per training epoch $B$, learning rate $\eta$, momentum $\vartheta$ \ \\
\textbf{Output}: Global model $\bf{w}$.\\
\textbf{Initialise}: Broadcast $\bf{w}^{(0)}$, Client selection $S_k = 0,\, \forall\, k \in [K]$
\begin{algorithmic}[1]
\For {$t = 0,\, 1,\, 2,\, \cdots,\, \lceil\frac{T_{\max}}{I_g} \rceil -1$}
\If {$t < \lceil \frac{K}{M} \rceil$}
\State $Q_t = \{M,\, (t+1)M,\, \cdots\}$ \Comment{Full pass on UEs}
\Else 
\State $Q_t = \max_Q\{\nu_k: \forall k\in Q_t\}$ \Comment{Greedy selection}
\EndIf

\For{client $k$ in $Q_t$}
\State $\bf{w}_{k}^{(t+1)} = $ ClientUpdate$(\mathcal{D}_k,\, \bf{w}^{(t)};\,E,\, B,\, \eta,\, \vartheta)$
\EndFor
\State ${\bf{w}}_y^{(t+1)}: y\in [N'] $ \Comment{Collected updates after RA}
\State $\bf{w}^{(t+1)}= $ModelAverage$(N_k,\, {\bf{w}}_k^{(t+1)}, {\bf{w}}_y^{(t+1)}: k \in Q_t)$
\State $\{\nu_k^{(t)}\}_{k \in Q_t} = $ GTG-Shapley$(\bf{w}^{(t)},\,\{\bf{w}_{k}^{(t+1)}\},\, \mathcal{D}_{\textrm{val}})$
\For{client $k$ in $Q_t$}
\Statex \hspace*{\algorithmicindent}\hspace*{\algorithmicindent}$\nu_k \gets \zeta\nu_k\, +\, (1 - \zeta)\nu_k^{(t)}$ \Comment{Exponential averaging}
\EndFor
\EndFor \\
\textbf{return} $\bf{w}^t$

\end{algorithmic}
}
\end{algorithm}
\vspace{-2\baselineskip}
\end{figure}

%% file: fig2_different_setting.tex
\begin{figure}[t!]
\centering
\begin{tikzpicture}

\definecolor{crimson2143940}{RGB}{214,39,40}
\definecolor{darkgray176}{RGB}{176,176,176}
\definecolor{darkorange25512714}{RGB}{255,127,14}
\definecolor{forestgreen4416044}{RGB}{44,160,44}
\definecolor{lightblue}{RGB}{173,216,230}
\definecolor{lightgray204}{RGB}{204,204,204}
\definecolor{lightgreen}{RGB}{144,238,144}
\definecolor{mediumpurple148103189}{RGB}{148,103,189}
\definecolor{steelblue31119180}{RGB}{31,119,180}
\definecolor{sienna1408675}{RGB}{140,86,75}

\begin{axis}[
width=0.8\linewidth,        
height=0.4\linewidth,   
scale only axis, 
legend cell align={left},
legend style={
  fill opacity=0.8,
  draw opacity=1,
  text opacity=1,
  at={(0.97,0.03)},
  anchor=south east,
  draw=lightgray204
},
tick align=outside,
tick pos=left,
x grid style={darkgray176},
xlabel={\# Global rounds $r$},
xmin=-10, xmax=210,
xtick style={color=black},
y grid style={darkgray176},
ylabel={Test accuracy},
ymin=0.0755000036209822, ymax=0.984099984541535,
ytick style={color=black}
]


\addplot [semithick, steelblue31119180]
table {%
0 0.350800007581711
1 0.568599998950958
2 0.721400022506714
3 0.777599990367889
4 0.809599995613098
5 0.834200024604797
6 0.852999985218048
7 0.86080002784729
8 0.870000004768372
9 0.873399972915649
10 0.884599983692169
11 0.883199989795685
12 0.88919997215271
13 0.889999985694885
14 0.897800028324127
15 0.899200022220612
16 0.898000001907349
17 0.901600003242493
18 0.904600024223328
19 0.901400029659271
20 0.902599990367889
21 0.902400016784668
22 0.899399995803833
23 0.903599977493286
24 0.90420001745224
25 0.897599995136261
26 0.910199999809265
27 0.908399999141693
28 0.908399999141693
29 0.90719997882843
30 0.91100001335144
31 0.916400015354156
32 0.910399973392487
33 0.91100001335144
34 0.907599985599518
35 0.916400015354156
36 0.913600027561188
37 0.917999982833862
38 0.909600019454956
39 0.912000000476837
40 0.917599976062775
41 0.916800022125244
42 0.91839998960495
43 0.916999995708466
44 0.919200003147125
45 0.918200016021729
46 0.921199977397919
47 0.920799970626831
48 0.920400023460388
49 0.921400010585785
50 0.924600005149841
51 0.922800004482269
52 0.922599971294403
53 0.923799991607666
54 0.917400002479553
55 0.92519998550415
56 0.925000011920929
57 0.92519998550415
58 0.925599992275238
59 0.924600005149841
60 0.925000011920929
61 0.924600005149841
62 0.92220002412796
63 0.925400018692017
64 0.919000029563904
65 0.923799991607666
66 0.923200011253357
67 0.922800004482269
68 0.923600018024445
69 0.92220002412796
70 0.923200011253357
71 0.924799978733063
72 0.926599979400635
73 0.925400018692017
74 0.925599992275238
75 0.921000003814697
76 0.923600018024445
77 0.926199972629547
78 0.927999973297119
79 0.923799991607666
80 0.924000024795532
81 0.926400005817413
82 0.929400026798248
83 0.92960000038147
84 0.928399980068207
85 0.923600018024445
86 0.92739999294281
87 0.923799991607666
88 0.930000007152557
89 0.925000011920929
90 0.927799999713898
91 0.927600026130676
92 0.928799986839294
93 0.926999986171722
94 0.926999986171722
95 0.928600013256073
96 0.930400013923645
97 0.930000007152557
98 0.930999994277954
99 0.929000020027161
100 0.929199993610382
101 0.930400013923645
102 0.929000020027161
103 0.929400026798248
104 0.930800020694733
105 0.928200006484985
106 0.933600008487701
107 0.932399988174438
108 0.930199980735779
109 0.926999986171722
110 0.930199980735779
111 0.930400013923645
112 0.934800028800964
113 0.932200014591217
114 0.933000028133392
115 0.934800028800964
116 0.931800007820129
117 0.93419998884201
118 0.934000015258789
119 0.933000028133392
120 0.930800020694733
121 0.93120002746582
122 0.932200014591217
123 0.930800020694733
124 0.932200014591217
125 0.933000028133392
126 0.931599974632263
127 0.929799973964691
128 0.933600008487701
129 0.934400022029877
130 0.933200001716614
131 0.931999981403351
132 0.935400009155273
133 0.933000028133392
134 0.932200014591217
135 0.932799994945526
136 0.934800028800964
137 0.934400022029877
138 0.935000002384186
139 0.93120002746582
140 0.934599995613098
141 0.931400001049042
142 0.930000007152557
143 0.937200009822845
144 0.935400009155273
145 0.936200022697449
146 0.934000015258789
147 0.932600021362305
148 0.934599995613098
149 0.935400009155273
150 0.934599995613098
151 0.934599995613098
152 0.931800007820129
153 0.931800007820129
154 0.935400009155273
155 0.934400022029877
156 0.935000002384186
157 0.934400022029877
158 0.933200001716614
159 0.935800015926361
160 0.932200014591217
161 0.93639999628067
162 0.936600029468536
163 0.930599987506866
164 0.937200009822845
165 0.933799982070923
166 0.932799994945526
167 0.935400009155273
168 0.934800028800964
169 0.929199993610382
170 0.932399988174438
171 0.934599995613098
172 0.933600008487701
173 0.931400001049042
174 0.935000002384186
175 0.936999976634979
176 0.935199975967407
177 0.937399983406067
178 0.938000023365021
179 0.936200022697449
180 0.937200009822845
181 0.936600029468536
182 0.93639999628067
183 0.934599995613098
184 0.933799982070923
185 0.937600016593933
186 0.937399983406067
187 0.935999989509583
188 0.937200009822845
189 0.935000002384186
190 0.936200022697449
191 0.934000015258789
192 0.931999981403351
193 0.936200022697449
194 0.935599982738495
195 0.93639999628067
196 0.935999989509583
197 0.935599982738495
198 0.938199996948242
199 0.932600021362305
};
\addlegendentry{Proposed (5/10)}
\addplot [semithick, darkorange25512714]
table {%
0 0.356400012969971
1 0.539799988269806
2 0.731599986553192
3 0.789200007915497
4 0.800599992275238
5 0.83899998664856
6 0.838599979877472
7 0.8682000041008
8 0.875599980354309
9 0.883800029754639
10 0.887399971485138
11 0.890200018882751
12 0.895600020885468
13 0.897800028324127
14 0.902599990367889
15 0.903400003910065
16 0.903800010681152
17 0.903800010681152
18 0.906799972057343
19 0.902199983596802
20 0.908800005912781
21 0.910799980163574
22 0.90939998626709
23 0.910000026226044
24 0.908599972724915
25 0.91159999370575
26 0.912800014019012
27 0.912999987602234
28 0.912000000476837
29 0.91540002822876
30 0.912800014019012
31 0.913800001144409
32 0.914799988269806
33 0.917400002479553
34 0.916800022125244
35 0.916999995708466
36 0.91839998960495
37 0.913800001144409
38 0.916400015354156
39 0.912999987602234
40 0.918799996376038
41 0.916999995708466
42 0.910199999809265
43 0.917800009250641
44 0.91839998960495
45 0.917999982833862
46 0.920400023460388
47 0.919600009918213
48 0.922399997711182
49 0.921999990940094
50 0.926800012588501
51 0.92220002412796
52 0.925800025463104
53 0.923399984836578
54 0.921800017356873
55 0.925800025463104
56 0.929400026798248
57 0.92739999294281
58 0.929000020027161
59 0.930400013923645
60 0.926999986171722
61 0.925400018692017
62 0.922800004482269
63 0.927600026130676
64 0.928399980068207
65 0.927200019359589
66 0.927200019359589
67 0.928399980068207
68 0.929000020027161
69 0.929199993610382
70 0.927600026130676
71 0.92739999294281
72 0.928799986839294
73 0.92739999294281
74 0.929000020027161
75 0.928399980068207
76 0.929400026798248
77 0.927600026130676
78 0.930800020694733
79 0.925800025463104
80 0.930800020694733
81 0.930599987506866
82 0.927999973297119
83 0.929799973964691
84 0.932200014591217
85 0.93120002746582
86 0.929799973964691
87 0.933399975299835
88 0.935400009155273
89 0.930800020694733
90 0.932799994945526
91 0.933600008487701
92 0.933399975299835
93 0.932600021362305
94 0.934400022029877
95 0.933200001716614
96 0.930599987506866
97 0.934599995613098
98 0.934599995613098
99 0.933000028133392
100 0.932399988174438
101 0.928399980068207
102 0.93419998884201
103 0.932399988174438
104 0.930800020694733
105 0.932799994945526
106 0.930599987506866
107 0.932799994945526
108 0.930800020694733
109 0.933200001716614
110 0.93120002746582
111 0.933399975299835
112 0.932200014591217
113 0.931400001049042
114 0.933200001716614
115 0.936999976634979
116 0.93639999628067
117 0.935800015926361
118 0.935800015926361
119 0.940199971199036
120 0.939400017261505
121 0.937799990177155
122 0.936999976634979
123 0.937399983406067
124 0.938799977302551
125 0.933600008487701
126 0.934400022029877
127 0.935400009155273
128 0.936600029468536
129 0.938799977302551
130 0.938799977302551
131 0.935999989509583
132 0.936800003051758
133 0.935400009155273
134 0.939400017261505
135 0.938399970531464
136 0.934400022029877
137 0.939400017261505
138 0.938000023365021
139 0.937200009822845
140 0.939000010490417
141 0.937200009822845
142 0.939999997615814
143 0.940199971199036
144 0.937799990177155
145 0.939400017261505
146 0.937399983406067
147 0.935000002384186
148 0.937799990177155
149 0.937200009822845
150 0.93639999628067
151 0.937799990177155
152 0.938199996948242
153 0.937799990177155
154 0.935599982738495
155 0.936800003051758
156 0.938399970531464
157 0.938199996948242
158 0.938399970531464
159 0.937799990177155
160 0.939800024032593
161 0.938399970531464
162 0.941200017929077
163 0.938000023365021
164 0.938000023365021
165 0.937399983406067
166 0.938199996948242
167 0.935199975967407
168 0.938000023365021
169 0.935599982738495
170 0.938199996948242
171 0.93639999628067
172 0.93860000371933
173 0.935800015926361
174 0.938000023365021
175 0.940400004386902
176 0.940599977970123
177 0.937600016593933
178 0.941999971866608
179 0.939000010490417
180 0.939400017261505
181 0.940599977970123
182 0.936999976634979
183 0.940199971199036
184 0.937799990177155
185 0.938799977302551
186 0.937399983406067
187 0.937399983406067
188 0.939000010490417
189 0.937600016593933
190 0.939800024032593
191 0.937399983406067
192 0.938199996948242
193 0.942799985408783
194 0.939999997615814
195 0.93860000371933
196 0.937399983406067
197 0.939199984073639
198 0.941399991512299
199 0.93860000371933
};
\addlegendentry{Proposed (10/10)}
\addplot [semithick, forestgreen4416044]
table {%
0 0.401600003242493
1 0.531599998474121
2 0.659600019454956
3 0.767599999904633
4 0.807600021362305
5 0.841600000858307
6 0.859399974346161
7 0.872600018978119
8 0.882799983024597
9 0.87279999256134
10 0.877799987792969
11 0.876999974250793
12 0.871399998664856
13 0.88239997625351
14 0.893800020217896
15 0.896000027656555
16 0.893599987030029
17 0.902400016784668
18 0.89819997549057
19 0.902199983596802
20 0.904600024223328
21 0.894999980926514
22 0.890999972820282
23 0.889599978923798
24 0.895799994468689
25 0.871999979019165
26 0.871800005435944
27 0.895600020885468
28 0.901799976825714
29 0.882000029087067
30 0.874800026416779
31 0.884999990463257
32 0.893800020217896
33 0.898000001907349
34 0.892000019550323
35 0.901799976825714
36 0.896200001239777
37 0.892799973487854
38 0.89740002155304
39 0.898000001907349
40 0.903599977493286
41 0.909200012683868
42 0.908999979496002
43 0.910600006580353
44 0.904600024223328
45 0.899999976158142
46 0.909200012683868
47 0.903800010681152
48 0.910399973392487
49 0.914200007915497
50 0.916400015354156
51 0.912599980831146
52 0.91100001335144
53 0.914399981498718
54 0.912199974060059
55 0.912199974060059
56 0.913999974727631
57 0.905200004577637
58 0.917200028896332
59 0.920199990272522
60 0.913600027561188
61 0.917999982833862
62 0.918200016021729
63 0.922800004482269
64 0.91619998216629
65 0.914200007915497
66 0.917800009250641
67 0.924000024795532
68 0.922399997711182
69 0.920400023460388
70 0.916999995708466
71 0.915600001811981
72 0.920799970626831
73 0.924000024795532
74 0.923200011253357
75 0.922999978065491
76 0.92220002412796
77 0.924000024795532
78 0.923799991607666
79 0.927600026130676
80 0.921599984169006
81 0.917999982833862
82 0.921000003814697
83 0.920400023460388
84 0.922800004482269
85 0.920400023460388
86 0.921400010585785
87 0.926599979400635
88 0.921599984169006
89 0.925800025463104
90 0.92519998550415
91 0.923600018024445
92 0.926199972629547
93 0.926199972629547
94 0.925800025463104
95 0.92220002412796
96 0.923600018024445
97 0.920199990272522
98 0.929199993610382
99 0.916999995708466
100 0.925800025463104
101 0.927200019359589
102 0.927799999713898
103 0.929799973964691
104 0.924399971961975
105 0.923799991607666
106 0.925800025463104
107 0.926800012588501
108 0.926599979400635
109 0.929000020027161
110 0.92739999294281
111 0.926199972629547
112 0.929000020027161
113 0.926999986171722
114 0.922599971294403
115 0.928799986839294
116 0.927200019359589
117 0.926199972629547
118 0.922399997711182
119 0.925000011920929
120 0.926400005817413
121 0.925999999046326
122 0.926199972629547
123 0.928399980068207
124 0.920199990272522
125 0.926800012588501
126 0.924199998378754
127 0.92960000038147
128 0.921199977397919
129 0.928799986839294
130 0.929000020027161
131 0.925400018692017
132 0.929199993610382
133 0.926599979400635
134 0.925000011920929
135 0.924199998378754
136 0.925599992275238
137 0.92059999704361
138 0.922800004482269
139 0.927799999713898
140 0.922800004482269
141 0.923399984836578
142 0.923399984836578
143 0.92220002412796
144 0.922800004482269
145 0.926599979400635
146 0.927600026130676
147 0.929400026798248
148 0.925999999046326
149 0.928399980068207
150 0.926800012588501
151 0.92739999294281
152 0.929199993610382
153 0.92220002412796
154 0.931999981403351
155 0.932399988174438
156 0.928399980068207
157 0.933600008487701
158 0.931400001049042
159 0.93120002746582
160 0.936200022697449
161 0.930199980735779
162 0.92739999294281
163 0.934599995613098
164 0.933399975299835
165 0.925800025463104
166 0.931400001049042
167 0.931400001049042
168 0.933200001716614
169 0.933399975299835
170 0.933000028133392
171 0.93120002746582
172 0.931400001049042
173 0.932399988174438
174 0.931599974632263
175 0.929400026798248
176 0.931800007820129
177 0.930199980735779
178 0.933799982070923
179 0.935999989509583
180 0.932399988174438
181 0.934800028800964
182 0.932600021362305
183 0.93120002746582
184 0.933399975299835
185 0.933600008487701
186 0.936600029468536
187 0.930199980735779
188 0.934599995613098
189 0.932799994945526
190 0.929400026798248
191 0.93120002746582
192 0.932799994945526
193 0.926599979400635
194 0.92960000038147
195 0.923600018024445
196 0.927200019359589
197 0.925000011920929
198 0.930400013923645
199 0.929799973964691
};
\addlegendentry{Proposed (20/10)}
\addplot [semithick, sienna1408675]
table {%
0 0.350800007581711
1 0.568599998950958
2 0.721400022506714
3 0.777599990367889
4 0.809599995613098
5 0.834200024604797
6 0.852999985218048
7 0.86080002784729
8 0.870000004768372
9 0.873399972915649
10 0.884599983692169
11 0.883199989795685
12 0.88919997215271
13 0.889999985694885
14 0.897800028324127
15 0.899200022220612
16 0.898000001907349
17 0.901600003242493
18 0.904600024223328
19 0.901400029659271
20 0.902599990367889
21 0.900799989700317
22 0.90039998292923
23 0.905600011348724
24 0.909799993038177
25 0.908200025558472
26 0.904999971389771
27 0.908800005912781
28 0.910199999809265
29 0.912999987602234
30 0.908599972724915
31 0.915600001811981
32 0.913800001144409
33 0.91540002822876
34 0.915600001811981
35 0.91540002822876
36 0.916599988937378
37 0.915799975395203
38 0.919600009918213
39 0.920199990272522
40 0.919000029563904
41 0.919200003147125
42 0.918600022792816
43 0.92059999704361
44 0.924799978733063
45 0.919000029563904
46 0.921400010585785
47 0.92220002412796
48 0.921800017356873
49 0.920400023460388
50 0.924399971961975
51 0.921999990940094
52 0.924799978733063
53 0.925599992275238
54 0.924399971961975
55 0.925400018692017
56 0.925999999046326
57 0.924399971961975
58 0.925000011920929
59 0.928399980068207
60 0.928600013256073
61 0.927600026130676
62 0.928600013256073
63 0.927999973297119
64 0.926999986171722
65 0.926999986171722
66 0.929000020027161
67 0.929799973964691
68 0.927999973297119
69 0.929000020027161
70 0.928200006484985
71 0.929799973964691
72 0.930599987506866
73 0.930800020694733
74 0.93120002746582
75 0.93120002746582
76 0.931999981403351
77 0.930999994277954
78 0.931999981403351
79 0.930000007152557
80 0.933600008487701
81 0.934800028800964
82 0.932399988174438
83 0.932799994945526
84 0.934599995613098
85 0.934000015258789
86 0.935000002384186
87 0.933200001716614
88 0.933200001716614
89 0.933600008487701
90 0.933200001716614
91 0.934000015258789
92 0.933000028133392
93 0.934800028800964
94 0.934800028800964
95 0.934800028800964
96 0.934599995613098
97 0.93419998884201
98 0.935599982738495
99 0.934000015258789
100 0.933600008487701
101 0.935400009155273
102 0.934000015258789
103 0.934400022029877
104 0.934000015258789
105 0.933399975299835
106 0.934599995613098
107 0.932600021362305
108 0.934400022029877
109 0.935199975967407
110 0.935400009155273
111 0.934400022029877
112 0.934400022029877
113 0.933799982070923
114 0.934400022029877
115 0.93419998884201
116 0.934400022029877
117 0.937200009822845
118 0.93419998884201
119 0.936999976634979
120 0.937600016593933
121 0.937399983406067
122 0.935199975967407
123 0.936600029468536
124 0.936999976634979
125 0.936800003051758
126 0.936600029468536
127 0.936600029468536
128 0.936800003051758
129 0.936999976634979
130 0.938000023365021
131 0.935999989509583
132 0.93639999628067
133 0.93860000371933
134 0.93639999628067
135 0.935000002384186
136 0.935400009155273
137 0.935999989509583
138 0.936999976634979
139 0.936800003051758
140 0.936600029468536
141 0.934599995613098
142 0.93639999628067
143 0.93639999628067
144 0.935599982738495
145 0.938199996948242
146 0.936999976634979
147 0.937799990177155
148 0.937200009822845
149 0.937399983406067
150 0.938399970531464
151 0.938199996948242
152 0.938199996948242
153 0.938799977302551
154 0.937200009822845
155 0.935599982738495
156 0.93860000371933
157 0.93860000371933
158 0.936800003051758
159 0.936999976634979
160 0.937200009822845
161 0.936999976634979
162 0.938399970531464
163 0.938399970531464
164 0.938199996948242
165 0.938799977302551
166 0.93860000371933
167 0.938000023365021
168 0.940599977970123
169 0.939800024032593
170 0.939199984073639
171 0.938399970531464
172 0.938199996948242
173 0.939000010490417
174 0.939000010490417
175 0.939800024032593
176 0.939999997615814
177 0.939199984073639
178 0.940199971199036
179 0.940999984741211
180 0.942200005054474
181 0.94080001115799
182 0.941600024700165
183 0.939800024032593
184 0.939400017261505
185 0.939800024032593
186 0.939400017261505
187 0.940400004386902
188 0.940599977970123
189 0.939400017261505
190 0.940400004386902
191 0.939800024032593
192 0.940199971199036
193 0.938799977302551
194 0.939400017261505
195 0.939199984073639
196 0.939599990844727
197 0.939199984073639
198 0.939199984073639
199 0.939000010490417
};
\addlegendentry{GreedyShap (20)}
\addplot [semithick, crimson2143940]
table {%
0 0.135399997234344
1 0.163000002503395
2 0.215399995446205
3 0.268200010061264
4 0.327199995517731
5 0.231800004839897
6 0.319799989461899
7 0.337000012397766
8 0.33500000834465
9 0.37279999256134
10 0.403400003910065
11 0.444000005722046
12 0.454600006341934
13 0.510399997234344
14 0.528599977493286
15 0.678200006484985
16 0.648000001907349
17 0.720399975776672
18 0.698199987411499
19 0.689000010490417
20 0.693599998950958
21 0.73580002784729
22 0.742200016975403
23 0.719799995422363
24 0.757200002670288
25 0.764999985694885
26 0.783800005912781
27 0.733799993991852
28 0.746399998664856
29 0.744799971580505
30 0.780600011348724
31 0.792999982833862
32 0.793600022792816
33 0.814000010490417
34 0.807200014591217
35 0.824199974536896
36 0.826200008392334
37 0.826200008392334
38 0.832000017166138
39 0.835399985313416
40 0.833199977874756
41 0.833000004291534
42 0.832400023937225
43 0.837999999523163
44 0.833999991416931
45 0.846599996089935
46 0.846800029277802
47 0.847000002861023
48 0.853799998760223
49 0.845200002193451
50 0.853200018405914
51 0.857599973678589
52 0.852199971675873
53 0.863200008869171
54 0.865599989891052
55 0.865400016307831
56 0.857599973678589
57 0.862600028514862
58 0.862200021743774
59 0.856599986553192
60 0.867200016975403
61 0.852999985218048
62 0.857800006866455
63 0.866199970245361
64 0.866199970245361
65 0.865000009536743
66 0.865000009536743
67 0.862999975681305
68 0.861599981784821
69 0.867399990558624
70 0.856199979782104
71 0.864799976348877
72 0.861999988555908
73 0.872200012207031
74 0.869599997997284
75 0.871200025081635
76 0.867799997329712
77 0.87279999256134
78 0.867200016975403
79 0.874599993228912
80 0.864799976348877
81 0.878400027751923
82 0.874599993228912
83 0.877399981021881
84 0.874800026416779
85 0.879199981689453
86 0.87720000743866
87 0.869000017642975
88 0.883400022983551
89 0.873000025749207
90 0.88239997625351
91 0.876200020313263
92 0.876399993896484
93 0.879199981689453
94 0.873399972915649
95 0.875999987125397
96 0.870199978351593
97 0.859600007534027
98 0.865199983119965
99 0.861199975013733
100 0.87279999256134
101 0.868799984455109
102 0.864000022411346
103 0.873600006103516
104 0.87279999256134
105 0.854200005531311
106 0.881399989128113
107 0.881200015544891
108 0.878600001335144
109 0.881399989128113
110 0.882799983024597
111 0.884999990463257
112 0.883000016212463
113 0.882000029087067
114 0.883000016212463
115 0.880200028419495
116 0.881200015544891
117 0.874000012874603
118 0.889800012111664
119 0.882600009441376
120 0.858799993991852
121 0.884599983692169
122 0.867399990558624
123 0.882200002670288
124 0.881399989128113
125 0.880200028419495
126 0.881399989128113
127 0.878400027751923
128 0.87720000743866
129 0.8817999958992
130 0.881200015544891
131 0.884400010108948
132 0.883000016212463
133 0.883800029754639
134 0.88620001077652
135 0.865800023078918
136 0.878199994564056
137 0.884199976921082
138 0.882200002670288
139 0.884400010108948
140 0.88400000333786
141 0.883800029754639
142 0.885800004005432
143 0.884999990463257
144 0.888000011444092
145 0.890600025653839
146 0.88239997625351
147 0.887600004673004
148 0.878600001335144
149 0.880400002002716
150 0.875999987125397
151 0.880400002002716
152 0.880800008773804
153 0.884599983692169
154 0.881399989128113
155 0.885599970817566
156 0.881600022315979
157 0.885999977588654
158 0.879800021648407
159 0.872200012207031
160 0.881600022315979
161 0.882000029087067
162 0.877399981021881
163 0.879800021648407
164 0.885399997234344
165 0.884800016880035
166 0.888199985027313
167 0.884800016880035
168 0.880200028419495
169 0.88919997215271
170 0.885599970817566
171 0.885599970817566
172 0.880200028419495
173 0.881600022315979
174 0.878799974918365
175 0.880999982357025
176 0.8817999958992
177 0.87720000743866
178 0.888000011444092
179 0.882200002670288
180 0.88620001077652
181 0.882200002670288
182 0.893599987030029
183 0.884800016880035
184 0.883199989795685
185 0.885800004005432
186 0.882799983024597
187 0.893800020217896
188 0.888400018215179
189 0.886600017547607
190 0.888199985027313
191 0.892199993133545
192 0.882799983024597
193 0.888199985027313
194 0.893400013446808
195 0.889999985694885
196 0.890799999237061
197 0.892199993133545
198 0.893199980258942
199 0.875199973583221
};
\addlegendentry{FedAvg (20)}
\addplot [semithick, mediumpurple148103189]
table {%
0 0.116800002753735
1 0.158199995756149
2 0.265199989080429
3 0.394199997186661
4 0.534200012683868
5 0.571600019931793
6 0.55620002746582
7 0.536400020122528
8 0.535000026226044
9 0.574999988079071
10 0.623600006103516
11 0.658399999141693
12 0.690999984741211
13 0.715399980545044
14 0.738200008869171
15 0.751999974250793
16 0.766200006008148
17 0.774399995803833
18 0.784200012683868
19 0.794799983501434
20 0.800000011920929
21 0.806599974632263
22 0.814599990844727
23 0.820200026035309
24 0.826600015163422
25 0.833999991416931
26 0.83840000629425
27 0.844200015068054
28 0.849399983882904
29 0.852800011634827
30 0.855599999427795
31 0.862800002098083
32 0.866999983787537
33 0.871999979019165
34 0.874199986457825
35 0.87720000743866
36 0.87720000743866
37 0.879999995231628
38 0.881399989128113
39 0.882200002670288
40 0.882799983024597
41 0.884999990463257
42 0.886799991130829
43 0.888000011444092
44 0.890399992465973
45 0.890600025653839
46 0.891600012779236
47 0.893400013446808
48 0.893400013446808
49 0.894800007343292
50 0.896200001239777
51 0.897000014781952
52 0.898400008678436
53 0.898800015449524
54 0.90039998292923
55 0.900600016117096
56 0.901000022888184
57 0.901199996471405
58 0.901000022888184
59 0.902400016784668
60 0.902599990367889
61 0.90200001001358
62 0.902599990367889
63 0.902800023555756
64 0.902800023555756
65 0.903800010681152
66 0.903800010681152
67 0.904799997806549
68 0.904999971389771
69 0.905600011348724
70 0.906000018119812
71 0.905799984931946
72 0.906199991703033
73 0.9064000248909
74 0.906599998474121
75 0.907000005245209
76 0.905799984931946
77 0.907400012016296
78 0.908200025558472
79 0.908999979496002
80 0.90939998626709
81 0.910199999809265
82 0.909799993038177
83 0.91100001335144
84 0.91100001335144
85 0.91159999370575
86 0.911199986934662
87 0.91159999370575
88 0.912599980831146
89 0.912599980831146
90 0.912999987602234
91 0.912800014019012
92 0.913399994373322
93 0.913600027561188
94 0.913800001144409
95 0.913800001144409
96 0.914600014686584
97 0.915199995040894
98 0.915199995040894
99 0.916000008583069
100 0.91540002822876
101 0.915600001811981
102 0.916800022125244
103 0.916999995708466
104 0.915600001811981
105 0.916999995708466
106 0.917200028896332
107 0.917599976062775
108 0.917400002479553
109 0.918200016021729
110 0.917400002479553
111 0.919000029563904
112 0.919399976730347
113 0.919000029563904
114 0.919000029563904
115 0.919200003147125
116 0.919200003147125
117 0.918600022792816
118 0.919600009918213
119 0.919000029563904
120 0.919399976730347
121 0.92059999704361
122 0.920199990272522
123 0.921000003814697
124 0.921000003814697
125 0.921199977397919
126 0.921400010585785
127 0.92220002412796
128 0.921800017356873
129 0.921599984169006
130 0.921999990940094
131 0.922599971294403
132 0.922399997711182
133 0.923399984836578
134 0.923399984836578
135 0.923200011253357
136 0.924199998378754
137 0.924399971961975
138 0.924399971961975
139 0.924799978733063
140 0.925400018692017
141 0.925599992275238
142 0.925400018692017
143 0.926400005817413
144 0.925999999046326
145 0.926599979400635
146 0.926800012588501
147 0.926999986171722
148 0.927200019359589
149 0.927200019359589
150 0.92739999294281
151 0.927600026130676
152 0.928399980068207
153 0.928600013256073
154 0.928600013256073
155 0.928600013256073
156 0.928799986839294
157 0.927999973297119
158 0.928600013256073
159 0.929400026798248
160 0.92960000038147
161 0.929799973964691
162 0.928799986839294
163 0.929799973964691
164 0.929799973964691
165 0.930199980735779
166 0.930599987506866
167 0.930599987506866
168 0.930400013923645
169 0.930999994277954
170 0.931599974632263
171 0.930800020694733
172 0.93120002746582
173 0.931999981403351
174 0.932399988174438
175 0.932399988174438
176 0.932399988174438
177 0.932399988174438
178 0.932600021362305
179 0.933399975299835
180 0.933600008487701
181 0.933799982070923
182 0.934000015258789
183 0.934000015258789
184 0.934800028800964
185 0.934599995613098
186 0.934400022029877
187 0.934400022029877
188 0.934599995613098
189 0.934400022029877
190 0.935000002384186
191 0.93419998884201
192 0.935000002384186
193 0.934599995613098
194 0.935000002384186
195 0.93639999628067
196 0.934599995613098
197 0.935999989509583
198 0.935999989509583
199 0.935199975967407
};
\addlegendentry{Centralized (20)}
\addplot [semithick, red, dashed]
table {%
20 0.0755000036209822
20 0.984099984541535
};
\addlegendentry{$r_{\mathrm{th}}=20$}
\end{axis}

\end{tikzpicture}
\vspace{-0.2cm}
\caption{Performance comparison under different \gls{ue} sampling and model aggregation approaches.} 
\label{fig:com_cost}
\vspace{-0.7cm}
\end{figure}
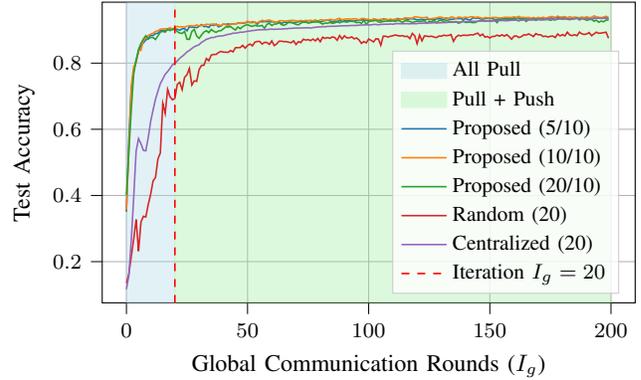

%% file: fig5_MNIST_acc_vs_no_users.tex
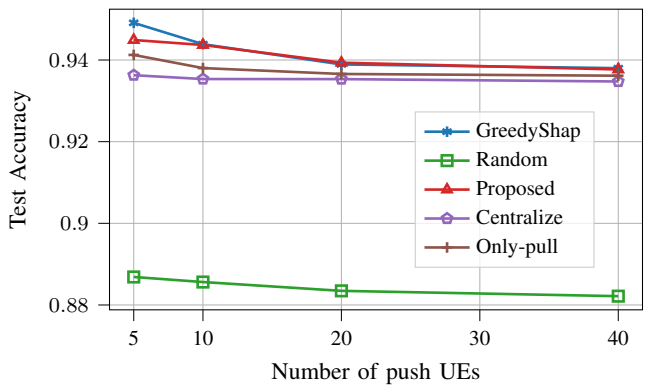
\begin{figure}[t!]
\begin{tikzpicture}

\definecolor{crimson2143940}{RGB}{214,39,40}
\definecolor{darkgray176}{RGB}{176,176,176}
\definecolor{darkorange25512714}{RGB}{255,127,14}
\definecolor{forestgreen4416044}{RGB}{44,160,44}
\definecolor{lightgray204}{RGB}{204,204,204}
\definecolor{mediumpurple148103189}{RGB}{148,103,189}
\definecolor{sienna1408675}{RGB}{140,86,75}
\definecolor{steelblue31119180}{RGB}{31,119,180}

\begin{axis}[
width=0.8\linewidth,        
height=0.4\linewidth,   
scale only axis,         
legend cell align={left},
legend style={
  fill opacity=0.8,
  draw opacity=1,
  text opacity=1,
  at={(0.91,0.4)},
  anchor=east,
  draw=lightgray204
},
tick align=outside,
tick pos=left,
xtick={50,100, 200, 300, 400},
xticklabels={5,10, 20, 30, 40},
x grid style={darkgray176},
xlabel={Number of push UEs $N$},
xmin=32.5, xmax=417.5,
xtick style={color=black},
y grid style={darkgray176},
ylabel={Test accuracy },
ymin=0.878813001811504, ymax=0.952447005808354,
ytick style={color=black}
]
\addplot [semithick, steelblue31119180, mark=asterisk, mark size=2, mark options={solid}, line width =1]
table {%
50 0.949100005626679
100 0.943879997730255
200 0.938919991254807
400 0.938
};
\addlegendentry{GreedyShap}
\addplot [semithick, forestgreen4416044, mark=square, mark size=2, mark options={solid}, line width =1]
table {%
50 0.886860001087189
100 0.885619992017746
200 0.88345999121666
400 0.882160001993179
};
\addlegendentry{FedAvg}
\addplot [semithick, crimson2143940, mark= triangle, mark size=2, mark options={solid}, line width =1]
table {%
50 0.944900000095367
100 0.943680000305176
200 0.939339995384216
400 0.937680006027222
};
\addlegendentry{Proposed}
\addplot [semithick, mediumpurple148103189, mark=pentagon, mark size=2, mark options={solid}, line width =1]
table {%
50 0.936280000209808
100 0.935339999198914
200 0.935320001840591
400 0.934739995002747
};
\addlegendentry{Centralized}
\addplot [semithick, sienna1408675, mark=+, mark size=2, mark options={solid}, line width =1]
table {%
50 0.941244436634911
100 0.93802222278383
200 0.936579996347427
400 0.936155564255185
};
\addlegendentry{Only pull}
\end{axis}

\end{tikzpicture}
\vspace{-0.2cm}
\caption{Test accuracy vs the number of push UEs per frame.} 
\label{fig:acc_vs_no_users}
\vspace{-0.5cm}
\end{figure}

%% file: fig6_latency_ver2.tex
\begin{tikzpicture}

\definecolor{darkgray176}{RGB}{176,176,176}
\definecolor{darkred17600}{RGB}{176,0,0}
\definecolor{forestgreen4416044}{RGB}{44,160,44}
\definecolor{lightgray204}{RGB}{204,204,204}
\definecolor{steelblue31119180}{RGB}{31,119,180}
\begin{axis}[
width=0.8\linewidth,        
height=0.4\linewidth,   
scale only axis, 
legend cell align={left},
legend style={
  fill opacity=0.8,
  draw opacity=1,
  text opacity=1,
  at={(0.03,0.97)},
  anchor=north west,
  draw=lightgray204
},
tick align=outside,
tick pos=left,
x grid style={darkgray176},
xlabel={Target accuracy $\theta_{\mathrm{th}}$},
xmin=-0.53, xmax=8.93,
xtick style={color=black},
xtick={0.2,1.2,2.2,3.2,4.2,5.2,6.2,7.2,8.2},
xticklabels={0.4,0.5,0.6,0.7,0.8,0.9,0.91,0.92,0.93},
y grid style={darkgray176},
ymin=0, ymax=15000,
ytick={5000, 10000, 15000},
yticklabels={5000, 10000, 15000},
ylabel style={font=\color{white!15!black}},
ylabel={\rv{Time cost [slots]}},
scaled ticks=false,
axis background/.style={fill=none},
]
\draw[draw=none,fill=forestgreen4416044] (axis cs:-0.1,0) rectangle (axis cs:0.1,2760);
\addlegendimage{ybar,ybar legend,draw=none,fill=forestgreen4416044}
\addlegendentry{FedAvg}

\draw[draw=none,fill=forestgreen4416044] (axis cs:0.9,0) rectangle (axis cs:1.1,3220);
\draw[draw=none,fill=forestgreen4416044] (axis cs:1.9,0) rectangle (axis cs:2.1,3680);
\draw[draw=none,fill=forestgreen4416044] (axis cs:2.9,0) rectangle (axis cs:3.1,4600);
\draw[draw=none,fill=forestgreen4416044] (axis cs:3.9,0) rectangle (axis cs:4.1,7130);
\draw[draw=none,fill=forestgreen4416044] (axis cs:4.9,0) rectangle (axis cs:5.1,0);
\draw[draw=none,fill=forestgreen4416044] (axis cs:5.9,0) rectangle (axis cs:6.1,0);
\draw[draw=none,fill=forestgreen4416044] (axis cs:6.9,0) rectangle (axis cs:7.1,0);
\draw[draw=none,fill=forestgreen4416044] (axis cs:7.9,0) rectangle (axis cs:8.1,0);
\draw[draw=none,fill=steelblue31119180] (axis cs:0.1,0) rectangle (axis cs:0.3,460);
\addlegendimage{ybar,ybar legend,draw=none,fill=steelblue31119180}
\addlegendentry{GreedyShap}

\draw[draw=none,fill=steelblue31119180] (axis cs:1.1,0) rectangle (axis cs:1.3,460);
\draw[draw=none,fill=steelblue31119180] (axis cs:2.1,0) rectangle (axis cs:2.3,690);
\draw[draw=none,fill=steelblue31119180] (axis cs:3.1,0) rectangle (axis cs:3.3,690);
\draw[draw=none,fill=steelblue31119180] (axis cs:4.1,0) rectangle (axis cs:4.3,1150);
\draw[draw=none,fill=steelblue31119180] (axis cs:5.1,0) rectangle (axis cs:5.3,3910);
\draw[draw=none,fill=steelblue31119180] (axis cs:6.1,0) rectangle (axis cs:6.3,5520);
\draw[draw=none,fill=steelblue31119180] (axis cs:7.1,0) rectangle (axis cs:7.3,8280);
\draw[draw=none,fill=steelblue31119180] (axis cs:8.1,0) rectangle (axis cs:8.3,14260);
\draw[draw=none,fill=darkred17600] (axis cs:0.3,0) rectangle (axis cs:0.5,460);
\addlegendimage{ybar,ybar legend,draw=none,fill=darkred17600}
\addlegendentry{Proposed}

\draw[draw=none,fill=darkred17600] (axis cs:1.3,0) rectangle (axis cs:1.5,460);
\draw[draw=none,fill=darkred17600] (axis cs:2.3,0) rectangle (axis cs:2.5,690);
\draw[draw=none,fill=darkred17600] (axis cs:3.3,0) rectangle (axis cs:3.5,690);
\draw[draw=none,fill=darkred17600] (axis cs:4.3,0) rectangle (axis cs:4.5,1150);
\draw[draw=none,fill=darkred17600] (axis cs:5.3,0) rectangle (axis cs:5.5,3910);
\draw[draw=none,fill=darkred17600] (axis cs:6.3,0) rectangle (axis cs:6.5,3910);
\draw[draw=none,fill=darkred17600] (axis cs:7.3,0) rectangle (axis cs:7.5,4859);
\draw[draw=none,fill=darkred17600] (axis cs:8.3,0) rectangle (axis cs:8.5,10283);
\end{axis}

\end{tikzpicture}

%% file: R1_RevisedManuscript.bbl
\begin{thebibliography}{10}

\bibitem{mcmahan2017communication}
B.~McMahan {\em et~al.}, ``Communication-efficient learning of deep networks
  from decentralized data,'' in {\em Artif. Intell. Statist.}, pp.~1273--1282,
  PMLR, 2017.

\bibitem{nguyen2021federated}
D.~C. Nguyen {\em et~al.}, ``Federated learning for internet of things: A
  comprehensive survey,'' {\em IEEE Commun. Surv. Tut.}, vol.~23, no.~3,
  pp.~1622--1658, 2021.

\bibitem{kountouris2021semantics}
M.~Kountouris and N.~Pappas, ``Semantics-empowered communication for networked
  intelligent systems,'' {\em IEEE Commun. Mag.}, vol.~59, no.~6, pp.~96--102,
  2021.

\bibitem{kairouz2021advances}
P.~Kairouz {\em et~al.}, ``Advances and open problems in federated learning,''
  {\em Found. Trends Mach Learn.}, vol.~14, no.~1--2, pp.~1--210, 2021.

\bibitem{singhal2024greedy}
P.~Singhal, S.~R. Pandey, and P.~Popovski, ``Greedy shapley client selection
  for communication-efficient federated learning,'' {\em IEEE Netw. Lett.},
  vol.~6, no.~2, pp.~134--138, 2024.

\bibitem{reisizadeh2022straggler}
A.~Reisizadeh {\em et~al.}, ``Straggler-resilient federated learning:
  Leveraging the interplay between statistical accuracy and system
  heterogeneity,'' {\em IEEE J. Sel. Areas Inf. Theory}, vol.~3, no.~2,
  pp.~197--205, 2022.

\bibitem{zhang2022federated}
T.~Zhang {\em et~al.}, ``Federated learning for the internet of things:
  Applications, challenges, and opportunities,'' {\em IEEE Internet Things
  Mag.}, vol.~5, no.~1, pp.~24--29, 2022.

\bibitem{pandey2020crowdsourcing}
S.~R. Pandey {\em et~al.}, ``A crowdsourcing framework for on-device federated
  learning,'' {\em IEEE Trans. Wireless Commun.}, vol.~19, no.~5,
  pp.~3241--3256, 2020.

\bibitem{cavallero2024coexistence}
S.~Cavallero {\em et~al.}, ``Coexistence of pull and push communication in
  wireless access for {IoT} devices,'' in {\em 2024 IEEE 25th Int. Workshop
  Signal Process. Adv. Wireless Commun. (SPAWC)}, pp.~841--845, 2024.

\bibitem{suman2023statistical}
S.~Suman {\em et~al.}, ``Statistical characterization of closed-loop latency at
  the mobile edge,'' {\em IEEE Trans. Commun.}, vol.~71, no.~7, pp.~4391--4405,
  2023.

\bibitem{konevcny2017semi}
J.~Kone{\v{c}}n{\`y}, Z.~Qu, and P.~Richt{\'a}rik, ``Semi-stochastic coordinate
  descent,'' {\em optimization Methods and Software}, vol.~32, no.~5,
  pp.~993--1005, 2017.

\bibitem{boucheron2003concentration}
S.~Boucheron, G.~Lugosi, and O.~Bousquet, ``Concentration inequalities,'' in
  {\em Summer school on machine learning}, pp.~208--240, Springer, 2003.

\bibitem{liu2022gtg}
Z.~Liu {\em et~al.}, ``{GTG}-shapley: Efficient and accurate participant
  contribution evaluation in federated learning,'' {\em ACM Trans. intell.
  Syst. Technol. (TIST)}, vol.~13, no.~4, pp.~1--21, 2022.

\bibitem{lecun1998gradient}
Y.~LeCun {\em et~al.}, ``Gradient-based learning applied to document
  recognition,'' {\em Proc. IEEE}, vol.~86, no.~11, pp.~2278--2324, 1998.

\bibitem{krizhevsky2009learning}
A.~Krizhevsky, ``Learning multiple layers of features from tiny images,'' {\em
  Tech. Rep.}, 2009.

\end{thebibliography}
